\documentclass[a4paper,11pt]{article}

\usepackage[sort&compress,square,comma,numbers]{natbib}
\usepackage{amsmath,amsfonts,amssymb,amsthm}
\usepackage{url}
\usepackage{color}
\usepackage{enumerate}
\usepackage{hyperref,pgf}
\usepackage{setspace}
\usepackage[all]{xy} % draw graphs
\usepackage{marginnote}
\usepackage[version=3]{mhchem}

\DeclareMathOperator{\im}{im} % like \ker

\newcommand{\R}{{\mathbb R}}
\newcommand{\Rnn}{\R_{\ge0}}
\newcommand{\Rp}{\R_{>0}}
\newcommand{\dd}[2]{\frac{\text{d} #1}{\text{d} #2}}
\newcommand{\lr}{\Leftrightarrow}

\newcommand{\LLL}{\mathcal{L}}
\newcommand{\AAA}{\mathcal{A}}
\newcommand{\GGG}{\mathcal{G}}
\newcommand{\PPP}{\mathcal{P}}

\newcommand{\SSS}{\mathcal{S}}
\newcommand{\CCC}{\mathcal{C}}
\newcommand{\RRR}{\mathcal{R}}

\newcommand{\TTT}{\mathcal{T}}
\newcommand{\EEE}{\mathcal{E}}

\newcommand{\cum}[1]{\breve{#1}}
\newcommand{\one}{1}
\newcommand{\con}{u}
\newcommand{\J}{{| \hspace{-1pt} J \hspace{-.5pt} |}}
\newcommand{\ts}{\textstyle}

\theoremstyle{plain}
\newtheorem{thm}{Theorem}
\newtheorem{pro}[thm]{Proposition}
\newtheorem{cor}[thm]{Corollary}
\newtheorem{lem}[thm]{Lemma}

\theoremstyle{definition}
\newtheorem{dfn}{Definition}

\newtheorem{exa}[dfn]{Example}
\newtheorem{sch}[dfn]{Scheme}

\makeatletter
\def\blfootnote{\xdef\@thefnmark{}\@footnotetext}
\makeatother

\begin{document}

\title{Genetic Recombination \\ as a Chemical Reaction Network}
%s with Mass-Action Kinetics

\author{Stefan M\"uller$^*$ and Josef Hofbauer}

\date{\today}

\maketitle

\begin{center}
\em \small
Dedicated to the memory of the Viennese chemists and mathematicians \\
Rudolf Wegscheider (1859--1935),
Hilda Geiringer (1893--1973), \\
and Friedrich J.\ M.\ Horn (1927--1978)
\end{center}

\begin{abstract}
\noindent
The process of genetic recombination can be seen as a chemical reaction network with mass-action kinetics.
We review the known results on existence, uniqueness, and global stability of an equilibrium
in every compatibility class and for all rate constants,
from both the population genetics and the reaction networks point of view.
% \smallskip
% 
% \noindent
% {\bf Keywords:}
% ???
% \smallskip
% 
% \noindent
% {\em Mathematics Subject Classification (2010):}
% ???
\end{abstract}

\blfootnote{
\scriptsize

\noindent
{\bf S.\ M\"uller} (\href{mailto:stefan.mueller@ricam.oeaw.ac.at}{stefan.mueller@ricam.oeaw.ac.at}).  
Johann Radon Institute for Computational and Applied Mathematics, Austrian Academy of Sciences, Altenbergerstra{\ss}e 69, 4040 Linz, Austria
\smallskip

\noindent
{\bf J.\ Hofbauer} (\href{mailto:josef.hofbauer@univie.ac.at}{josef.hofbauer@univie.ac.at}). 
Department of Mathematics, University of Vienna, Oskar-Morgenstern-Platz 1, 1090 Wien, Austria 
\smallskip

\noindent
$^*$ Corresponding author
}

%%%%%%%%%%%%%%%%%%%%%%%%%%%%%%%%%%%%%%%%%%%%%%%%%%%%%%%%%%%%%%%%%%%%%%%%%%%%%%%%%%%%%%%%%%%%%%%%%%%%%%%%%%%%%%%%%%%%%%%%%%

\section{Introduction}

Recombination, or chromosomal crossover, is the exchange of genetic material between a pair of homologous chromosomes.
It occurs when matching regions on matching chromosomes break and then reconnect to the other chromosome.
Recombination is the major evolutionary force to produce and maintain variation in a (sexual) population.

The process of recombination
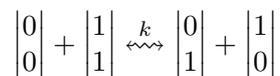
\[
\begin{vmatrix}
0 \\ 0
\end{vmatrix}
+
\begin{vmatrix}
1 \\ 1
\end{vmatrix}
\overset{k}{\leftrightsquigarrow}
\begin{vmatrix}
0 \\ 1
\end{vmatrix}
+
\begin{vmatrix}
1 \\ 0
\end{vmatrix}
\]
involves two loci, two alleles per locus, and hence four gametes.
It can be written as a reversible chemical reaction
\[
g_1 + g_4 \overset{k}{\lr} g_2 + g_3 .
\]
This is a very simple reaction network with deficiency zero.
The recombination $g_1 + g_4 \to g_2 + g_3$ occurs at the rate $k \, p_1 p_4$
determined by the rate constant times the frequencies of the reacting gametes.
In the chemical setting, this corresponds to the assumption of mass-action kinetics.
The dynamical system for the gamete frequencies amounts to
\[
\dot p_1 = \dot p_4 = k \, (-p_1 p_4 + p_2 p_3) = - \dot p_2 = - \dot p_3 .
\]
There are conservation laws
\[
(p_1 + p_2)^{\ts \cdot} = (p_1 + p_3)^{\ts \cdot}  = (p_4 + p_2)^{\ts \cdot} = (p_4 + p_3)^{\ts \cdot} = 0
\]
of which three are linearly independent.
The equilibrium manifold is given by the conic
\[
p_1 p_4 = p_2 p_3 .
\]
In each stoichiometric compatibility class,
solutions converge to a unique detailed-balancing equilibrium.
This is usually proved by considering the so-called linkage disequilibrium function
$D = p_1 p_4 - p_2 p_3$, which satisfies
\[
\dot D = - k \, D \, \left( p_1 + p_2 + p_3 + p_4 \right) .
\]
Hence, $D$ converges to $0$.

Alternatively, one can try an ansatz for a Lyapunov function in the form
$V(p) = \sum_{i=1}^4 F(p_i)$.
Then,
\[
\dot V(p) =  \sum_{i=1}^4 F'(p_i) \, \dot p_i = - k \, D \left( F'(p_1) + F'(p_4) - F'(p_2) - F'(p_3) \right) .
\]
For the choice $F'(p) = \ln p$, we obtain
\begin{align*}
\dot V(p) &= - k \, D \left( \ln p_1 + \ln p_4 - \ln p_2 - \ln p_3 \right) \\
&= - k \, (p_1 p_4 - p_2 p_3) \left( \ln (p_1 p_4) - \ln (p_2 p_3) \right) \leq 0
\end{align*}
due to the monotonicity of the logarithm.
This shows that $V(p)$, with the convex function $F(p) = p \ln p - p$, is a global Lyapunov function.

The main object of this paper is to study the general recombination
model in continuous time, with an arbitrary number of genetic loci and arbitrary numbers of alleles at each locus.
We will see that this leads to a chemical reaction network which is reversible,
satisfies the Wegscheider conditions~\cite{Wegscheider1901} (since the rate constants of a reaction and its reverse coincide),
and is detailed-balancing.
A generalization of the above entropy-like Lyapunov function allows us to prove global stability.

In population genetics, this general recombination model was studied (in discrete time) by Geiringer~\cite{geiringer1944},
and further by~\cite{shah1979, nagylaki1993, nhb1999, buerger2000}.
Their proofs use linkage disequilibrium functions, induction on the number of loci, cumulants, etc., and are far from easy.
Simpler proofs based on the entropy as Lyapunov function were independently given by Akin~\cite{akin1979} (in continuous time)
and Lyubich~\cite{lyubich1992} (in discrete time).

%%%%%%%%%%%%%%%%%%%%%%%%%%%%%%%%%%%%%%%%%%%%%%%%%%%%%%%%%%%%%%%%%%%%%%%%%%%%%%%%%%%%%%%%%%%%%%%%%%%%%%%%%%%%%%%%%%%%%%%%%%

\section{Mathematical model}

{\bf Notation:} We denote the positive real numbers by $\Rp$ and the non-negative real numbers by $\Rnn$.
For a finite index set $I$, we write $\R^I$ for the real vector space of formal sums $x = \sum_{i \in I} x_i \, i$
with $x_i \in \R$.
Viewing the elements of $I$ as indicator functions,
a vector $x \in \R^I$ can be seen as a function $x \colon I \to \R$, and $x(i) = x_i$.
%Analogously, we write $\Rp^I$ and $\Rnn^I$ for the corresponding subsets.
For $x,y \in \Rnn^I$, we define $x^y \in \Rnn$ as $\prod_{i \in I} x_i^{y_i}$, where we set $0^0=1$.
Given a matrix $Y \in \R^{I \times J}$, we denote by $Y^j \in \R^I$ the column vector indexed by $j \in J$.
For $x \in \Rnn^I$ and $Y \in \Rnn^{I \times J}$, we define $x^Y \in \Rnn^J$ as
$(x^Y)_j = x^{Y^j} = \prod_{i \in I} x_i^{Y_{ij}}$
for $j \in J$.

\subsection{Genetic recombination}

We consider a finite set of loci $\LLL$ with $L = |\LLL| \ge 1$,
finite sets of alleles $\AAA_i$ with $A_i = |\AAA_i| \ge 2$ for $i \in \LLL$,
the resulting set of gametes
\[
\GGG = \AAA_1 \times \ldots \times \AAA_L ,
\]
and the set of recombination patterns
\[
\PPP = \{ \{I,J\} \mid I \subseteq \LLL, \, J = \LLL \setminus I \} .
\]
In a recombination following pattern $\{I,J\}$,
alleles at loci $I$ get separated from alleles at loci $J$.
There are $|\GGG| = \prod_{i \in \LLL} A_i$ gametes
and $|\PPP|= 2^{L-1}$ recombination patterns,
including the trivial recombination $\{\emptyset,\LLL\}$.

Further,
we introduce the distribution of gamete frequencies
\[
p \colon \GGG \to \Rnn
\]
% with $\sum_{g \in \GGG} p(g) = 1$
and the distribution of recombination rate constants
\[
c \colon \PPP \to \Rnn .
\]
We identify the function $p \colon \GGG \to \Rnn$ with the vector $p \in \Rnn^\GGG$ and write $p = \sum_{g \in \GGG} p(g) \, g$.
Usually, we are interested in elements of the simplex
\[
S_\GGG = \{ p \in \Rnn^\GGG \mid \ts \sum_{g \in \GGG} p(g) = 1\} .
\]

For gametes $g,h \in \GGG$ and a recombination pattern $\{I,J\} \in \PPP$,
we define $g_I h_J \in \GGG$ as
\[
(g_I h_J)_i =
\begin{cases}
g_i, \text{ if } i \in I, \\
h_i, \text{ if } i \in J
\end{cases}
\]
and the resulting recombination as
\begin{equation} \label{eq:recombination}
\{g,h\}
\overset{c(\{I,J\})}{\rightsquigarrow}
\{ g_I h_J , g_J h_I \}
\end{equation}
with rate constant $c(\{I,J\})$.
For $g_I h_J, \, g_J h_I \in \GGG$ and $\{I,J\} \in \PPP$,
we find
\[
(g_I h_J)_I (g_J h_I)_J = g \quad \text{and} \quad  (g_I h_J)_J (g_J h_I)_I = h
\]
and obtain the reverse recombination
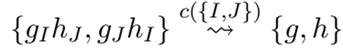
\[
\{ g_I h_J , g_J h_I \}
\overset{c(\{I,J\})}{\rightsquigarrow}
\{g,h\}
\]
which occurs with the same rate constant.

Clearly, recombination \eqref{eq:recombination} causes a change in gamete frequencies
only if $\{g,h\} \neq \{g_I h_J,g_J h_I\}$ and $c(\{I,J\}) > 0$.
Further, different recombination patterns may give rise to the same recombination (with different rate constants, in general).
In order to view recombination as a chemical reaction,
we have to ensure inequality of left- and right-hand sides and positivity of rate constants.
Moreover, we have to sum over the rate constants of all contributing recombination patterns
which can be seen as reaction mechanisms.

\subsection{Chemical reactions} \label{subsec:chemreac}

Let $K \subseteq \LLL$ be a subset of loci.
The recombination pattern $\{I,J\} \in \PPP$ induces the subpattern $\{I,J\}_K \in \PPP_K$
where $\{I,J\}_K = \{I \cap K, J \cap K \}$ and
\[
\PPP_K = \{ \{I,J\} \mid I \subseteq K, \, J=K \setminus I \} .
\]
We write $\{I,J\} \ge \{I,J\}_K$ and, for simplicity, $\PPP_K^* = \PPP_K \setminus \{ \{ \emptyset, K \} \}$.
The set of all recombination subpatterns amounts to
\[
\cum{\PPP} = \bigcup_{K \subseteq \LLL} \PPP_K ,
\]
and we introduce the distribution of cumulative recombination rate constants
\begin{align*}
\cum{c} \colon & \cum{\PPP} \to \Rnn , \\
& \{I,J\} \mapsto \sum_{\substack{ \{I',J'\} \in \PPP : \\ \{I',J'\} \ge \{I,J\} }} c(\{I',J'\}) .
\end{align*}
For a subpattern $\{I,J\}$ with $J=K \setminus I$ and $K \subseteq \LLL$,
the cumulative rate constant $\cum{c}(\{I,J\})$ 
sums over all patterns which agree with $\{I,J\}$ on~$K$.

An important parameter is the cumulative rate constant for the recombination subpattern $\{ \{i\},\{j\} \}$,
that is, for the case that an allele at locus $i$ gets separated from an allele at locus $j$.
We assume that $\cum{c}(\{ \{i\},\{j\} \})>0$ for all pairs of loci $i,j$.
Otherwise, the two loci can be identified.

To explicitly state a chemical reaction
arising from a recombination pattern and a pair of gametes,
we introduce the set $\Delta(g,h) = \{ i \in \LLL \mid g_i \neq h_i \}$ for gametes $g,h \in \GGG$.
In genetic terms, $g$ and $h$ are heterozygous at the subset of loci $\Delta(g,h)$ and homozygous otherwise.

Now, gametes $g,h \in \GGG$ and a recombination pattern $\{I,J\} \in \PPP$ give rise to a reaction mechanism,
only if $|\Delta(g,h)| \ge 2$, $\{I,J\}_{\Delta(g,h)} \neq \{\emptyset,\Delta(g,h)\}$, and $c(\{I,J\}) > 0$.
In other words, only if the gametes are heterozygous at two or more loci,
if the subpattern induced on these loci is non-trivial,
and if the recombination rate constant is non-zero.
Every pattern $\{I',J'\} \in \PPP$ with $\{I',J'\} \ge \{I,J\}_{\Delta(g,h)}$ and $c(\{I',J'\}) > 0$
gives rise to a mechanism for the same reaction,
that is, to the same recombination (with different rate constant, in general).
The effect of all such patterns is summarized in the chemical reaction
\begin{equation} \label{eq:reaction}
g + h
\overset{k}{\to}
g_I h_J + g_J h_I
\end{equation}
with rate constant $k \equiv k(g + h \to g_I h_J + g_J h_I) = \cum{c}(\{I,J\}_{\Delta(g,h)}) > 0$.
Note that $g + h$ stands for $\{g,h\}$ such that $g + h$ equals $h + g$.
The reverse reaction
\[
g_I h_J + g_J h_I
\overset{k}{\to}
g + h
\]
occurs with the same rate constant.

\subsection{Reaction networks}

A chemical reaction network $(\SSS,\CCC,\RRR)$ consists of three finite sets:
a set $\SSS$ of species, a set $\CCC \subset \Rnn^\SSS$ of complexes,
and a set $\RRR \subset \CCC \times \CCC$ of reactions.
Complexes are the left- and right-hand sides of reactions.
A complex $y \in \CCC$ can be seen as a formal sum of species $y = \sum_{s \in \SSS} y_s \, s$,
where $y_s$ is the stoichiometric coefficient of species $s$.
For a reaction $(y,y') \in \RRR$, we write $y \to y'$.
It is required that each complex appears in at least one reaction
and that there are no reactions of the form $y \to y$.

A chemical reaction network $(\SSS,\CCC,\RRR)$
together with a vector of rate constants $k \in \Rp^\RRR$
gives rise to a weighted directed graph with complexes as nodes, reactions as edges, and rate constants as labels.
The connected components of this graph are called linkage classes.
(Note that linkage classes have nothing to do with genetic linkage.)
A network is called weakly reversible
if every component is strongly connected,
that is, if there exists a directed path from each node to every other node in the component.

In the process of genetic recombination, the reacting species are the gametes, that is, $\SSS=\GGG$.
Every complex $g+h$ is a formal sum (with stoichiometric coefficients equal to one)
of two gametes $g$ and~$h$, which differ at two or more loci,
and every reaction $g+h \to g_I h_J + g_J h_I$ arises from a pair of gametes and a recombination pattern $\{I,J\}$,
under the conditions specified in the previous subsection.
The set of all chemical reactions (with corresponding rate constants) amounts to
\begin{align} \label{eq:R}
\RRR = \Big\{ g+h \overset{k}{\to} g_I h_J + g_J h_I \, \Big| & \,
g,h \in \GGG , \, \{I,J\} \in \PPP\text{ with } |\Delta(g,h)| \ge 2, \\ & \nonumber
\{I,J\}_{\Delta(g,h)} \neq \{\emptyset,\Delta(g,h)\}, \text{ and} \\& \nonumber
k \equiv \cum{c}(\{I,J\}_{\Delta(g,h)}) > 0 \Big\} .
\end{align}
For each reaction $y \to y'$ we have its reverse $y' \to y$, and both occur with the same rate constant.
Hence, we can combine them in the reversible reaction $y \lr y'$, which we identify with $y' \lr y$,
and write $k(y \lr y')$ for $k(y \to y') = k(y' \to y)$.
From~\eqref{eq:R},
we obtain the set of all reversible reactions
\begin{equation} \label{eq:RR}
\RRR_\lr = \{ y \overset{k}{\lr} y' \mid (y \overset{k}{\to} y') \in \RRR \}
\end{equation}
and the set of all complexes
\begin{equation} \label{eq:C}
\CCC = \{ y \mid (y \to y') \in \RRR \} .
% = \{ y \mid (y \lr y') \in \RRR_\lr \} .
\end{equation}

In the examples and schemes below, we determine the set of all (reversible) reactions in another way.
We first iterate over subsets of two or more loci
and then over non-trivial subpatterns on these loci:
\begin{align*} 
\RRR_\lr = \Big\{ g+h \overset{k}{\lr} g_I h_J + g_J h_I \, \Big| & \,
K \subseteq \LLL \text{ with } |K| \ge 2, \\ &
g,h \in \GGG \text{ with } |\Delta(g,h)| = K, \, \\ &
\{I,J\} \in \PPP_{\Delta(g,h)}^* \text{ with } k \equiv \cum{c}(\{I,J\}) > 0 \Big\} .
\end{align*}
Thereby, we extend the definition of $g_I h_J$ to the subpattern $\{I,J\} \in \PPP_{\Delta(g,h)}$
in the obvious way:
$(g_I h_J)_i = g_i$ for $i \in I$,
$(g_I h_J)_i = h_i$ for $i \in J$, and
$(g_I h_J)_i = g_i = h_i$ for $i \in \LLL \setminus (I \cup J)$.

Finally, we consider the graph arising from the reaction network $(\SSS,\CCC,\RRR)$,
in particular, its linkage classes.
We observe that species (gametes) consist of alleles
and complexes (pairs of gametes) contain two alleles at each locus.
Since reactions separate alleles, but do not consume or produce them,
only complexes which contain the same alleles are connected by a reaction.
Moreover, if complexes $g+h$ and $g'+h'$ are connected by a reaction
then $\Delta(g,h) = \Delta(g',h')$,
and every subpattern $\{I,J\} \in \PPP_{\Delta(g,h)}^*$ which gives rise to a reaction involving $g+h$
gives rise to a reaction involving $g'+h'$, and vice versa.
Hence every linkage class is a symmetric graph.
If no reaction is precluded by a zero rate constant,
then every linkage class is a complete graph,
characterized by two (possibly identical) alleles at each locus. 

\subsection{Examples and schemes}

We consider examples of genetic recombination for small numbers of loci and alleles
and depict the corresponding chemical reaction networks as graphs.
Further, we present schemes for arbitrary numbers of loci
and compute the resulting numbers of linkage classes, complexes, and reversible reactions.
For simplicity, we assume that no reaction is precluded by a zero rate constant.
In this case, all linkage classes are complete graphs.

Instead of $c(\{I,J\})$ we write $c(I)$ and further omit the set brackets,
e.g., $c(\{\{1\},\LLL \setminus \{1\}\}) \equiv c(\{1\}) \equiv c(1)$.

\begin{exa}[$L=2$ loci with $A_1=A_2=2$ alleles]
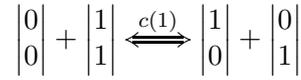
\[
\xymatrix{
{\begin{vmatrix}0\\0\end{vmatrix}+\begin{vmatrix}1\\1\end{vmatrix}}
\ar@{<=>}[r]^{c(1)} &
{\begin{vmatrix}1\\0\end{vmatrix}+\begin{vmatrix}0\\1\end{vmatrix}}
}
\]
The graph has $l=1$ linkage class, $m=2$ complexes, and $r=1$ reversible reaction.
\end{exa}

\begin{exa}[$L=3$ loci with $A_1=A_2=A_3=2$ alleles]
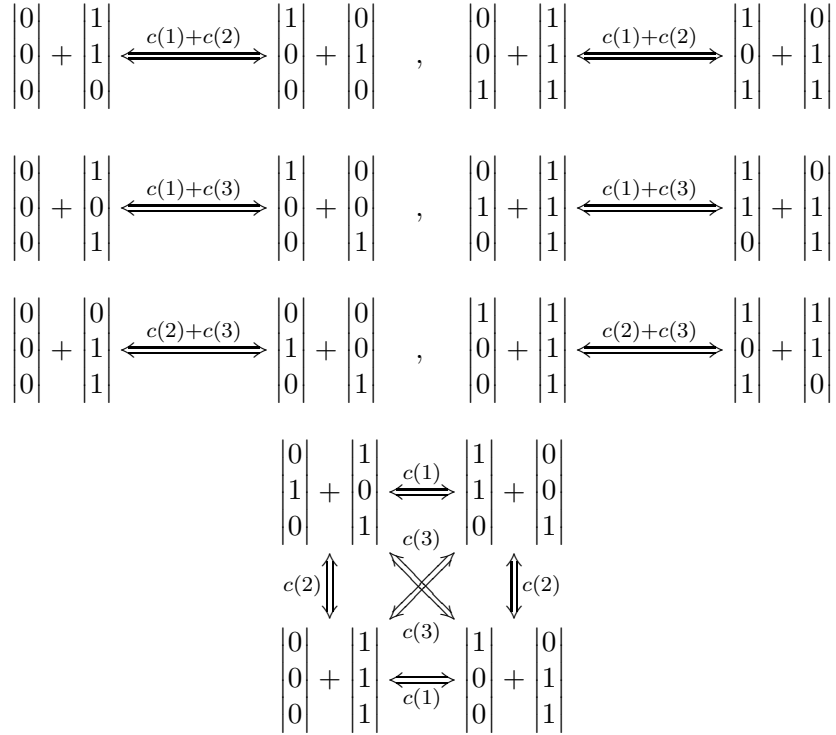
\[
\xymatrix{
{\begin{vmatrix}0\\0\\0\end{vmatrix}+\begin{vmatrix}1\\1\\0\end{vmatrix}}
\ar@{<=>}[rr]^{c(1)+c(2)} & &
{\begin{vmatrix}1\\0\\0\end{vmatrix}+\begin{vmatrix}0\\1\\0\end{vmatrix}}
}
\quad , \quad
\xymatrix{
{\begin{vmatrix}0\\0\\1\end{vmatrix}+\begin{vmatrix}1\\1\\1\end{vmatrix}}
\ar@{<=>}[rr]^{c(1)+c(2)} & &
{\begin{vmatrix}1\\0\\1\end{vmatrix}+\begin{vmatrix}0\\1\\1\end{vmatrix}}
}
\]
\[
\xymatrix{
{\begin{vmatrix}0\\0\\0\end{vmatrix}+\begin{vmatrix}1\\0\\1\end{vmatrix}}
\ar@{<=>}[rr]^{c(1)+c(3)} & &
{\begin{vmatrix}1\\0\\0\end{vmatrix}+\begin{vmatrix}0\\0\\1\end{vmatrix}}
}
\quad , \quad
\xymatrix{
{\begin{vmatrix}0\\1\\0\end{vmatrix}+\begin{vmatrix}1\\1\\1\end{vmatrix}}
\ar@{<=>}[rr]^{c(1)+c(3)} & &
{\begin{vmatrix}1\\1\\0\end{vmatrix}+\begin{vmatrix}0\\1\\1\end{vmatrix}}
}
\]
\[
\xymatrix{
{\begin{vmatrix}0\\0\\0\end{vmatrix}+\begin{vmatrix}0\\1\\1\end{vmatrix}}
\ar@{<=>}[rr]^{c(2)+c(3)} & &
{\begin{vmatrix}0\\1\\0\end{vmatrix}+\begin{vmatrix}0\\0\\1\end{vmatrix}}
}
\quad , \quad
\xymatrix{
{\begin{vmatrix}1\\0\\0\end{vmatrix}+\begin{vmatrix}1\\1\\1\end{vmatrix}}
\ar@{<=>}[rr]^{c(2)+c(3)} & &
{\begin{vmatrix}1\\0\\1\end{vmatrix}+\begin{vmatrix}1\\1\\0\end{vmatrix}}
}
\]
\[
\xymatrix{
{\begin{vmatrix}0\\1\\0\end{vmatrix}+\begin{vmatrix}1\\0\\1\end{vmatrix}}
\ar@{<=>}[r]^{c(1)} \ar@{<=>}[d]_{c(2)} \ar@{<=>}[rd]^(0.35){c(3)} &
{\begin{vmatrix}1\\1\\0\end{vmatrix}+\begin{vmatrix}0\\0\\1\end{vmatrix}}
\ar@{<=>}[d]^{c(2)} \ar@{<=>}[ld]^(0.65){c(3)} \\
{\begin{vmatrix}0\\0\\0\end{vmatrix}+\begin{vmatrix}1\\1\\1\end{vmatrix}}
\ar@{<=>}[r]_{c(1)} &
{\begin{vmatrix}1\\0\\0\end{vmatrix}+\begin{vmatrix}0\\1\\1\end{vmatrix}}
}
\]
The graph has $l=7$ linkage classes, $m=16$ complexes, and $r=12$ reversible reactions.
% Note that all linkage classes are complete graphs.
The last class has $2^{L-1}=4$ complexes and $\binom{2^{L-1}}{2} = \binom{4}{2} = 6$ reactions.
\end{exa}

\begin{sch}[$L\ge2$ loci with $A_i=2$ alleles, $i=1,\ldots,L$]
\begin{align*}
l &= \sum_{k=2}^L \binom{L}{k} \, 2^{L-k} \\
&= 3^L - (2^L + L \, 2^{L-1}) \\
&= 3^L - 2^{L-1} (2+L)
\end{align*}
\begin{align*}
m &= \sum_{k=2}^L \binom{L}{k} \, 2^{L-k} \, 2^{k-1} \\
&= \ts 2^{L-1} \sum_{k=2}^L \binom{L}{k} \\
&= 2^{L-1} (2^L - (1+L))
\end{align*}
\begin{align*}
r &= \sum_{k=2}^L \binom{L}{k} \, 2^{L-k} \, \binom{2^{k-1}}{2} \\
&= \ts \sum_{k=2}^L \binom{L}{k} \, 2^{L-k} \, 2^{k-1} (2^{k-1} - 1) \, 2^{-1} \\
&= \ts \sum_{k=2}^L \binom{L}{k} \, (2^{L-3} \, 2^k - 2^{L-2}) \\
&= 2^{L-3} (3^L - (1 + L \, 2)) - 2^{L-2} (2^L - (1+L)) \\
&= 2^{L-3} (3^L - 1) - 2^{L-2} (2^L - 1) % = 2^{L-3} (3^L - 2^{L+1} + 1)
\end{align*}
\end{sch}

\begin{exa}[$L=2$ loci with $A_1=2$ and $A_2=3$ alleles, respectively]
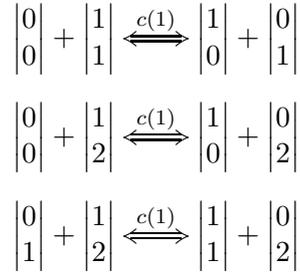
\begin{gather*}
\xymatrix{
{\begin{vmatrix}0\\0\end{vmatrix}+\begin{vmatrix}1\\1\end{vmatrix}}
\ar@{<=>}[r]^{c(1)} &
{\begin{vmatrix}1\\0\end{vmatrix}+\begin{vmatrix}0\\1\end{vmatrix}}
} \\
\xymatrix{
{\begin{vmatrix}0\\0\end{vmatrix}+\begin{vmatrix}1\\2\end{vmatrix}}
\ar@{<=>}[r]^{c(1)} &
{\begin{vmatrix}1\\0\end{vmatrix}+\begin{vmatrix}0\\2\end{vmatrix}}
} \\
\xymatrix{
{\begin{vmatrix}0\\1\end{vmatrix}+\begin{vmatrix}1\\2\end{vmatrix}}
\ar@{<=>}[r]^{c(1)} &
{\begin{vmatrix}1\\1\end{vmatrix}+\begin{vmatrix}0\\2\end{vmatrix}}
}
\end{gather*}
The graph has $l=3$ linkage classes, $m=6$ complexes, and $r=3$ reversible reactions.
\end{exa}

\begin{sch}[$L\ge2$ loci with $A_i\ge2$ alleles, $i=1,\ldots,L$]
\begin{align*}
l &= \sum_{K \subseteq \LLL : |K|\ge2} \,
\prod_{i \in K} \binom{A_i}{2} \prod_{i \in \LLL \setminus K} A_i \\
m &= \sum_{K \subseteq \LLL : |K|\ge2} \,
\prod_{i \in K} \binom{A_i}{2} \prod_{i \in \LLL \setminus K} A_i \,\, 2^{|K|-1} \\
r &= \sum_{K \subseteq \LLL : |K|\ge2} \,
\prod_{i \in K} \binom{A_i}{2} \prod_{i \in \LLL \setminus K} A_i \, \binom{2^{|K|-1}}{2} 
\end{align*}
\end{sch}

% Observation: Every species appears in a number of linkage classes:
% \[
% \frac{l \cdot n_*}{n} = \frac{\prod_i \binom{A_i}{2} \cdot 2^L}{\prod_i A_i}
% = \ts \prod_i (A_i-1)
% \]

\subsection{Dynamics}

Recombination~\eqref{eq:recombination} causes a change in gamete frequencies proportional to $g_I h_J + g_J h_I - g - h$
at the rate $c(\{I,J\}) \, p(g) \, p(h)$
determined by the recombination rate constant times the frequencies of the ``reacting'' gametes.
We formulate the dynamical system for the vector $p \in \Rnn^\GGG$ of all gamete frequencies,
that is, for $p = \sum_{g \in \GGG} p(g) \, g$,
by summing over all recombination partners and patterns:
\begin{equation} \label{eq:dynIJ}
\dd{p}{t} = \frac{1}{2} \sum_{g,h \in \GGG} \, \sum_{\{I,J\} \in \PPP} c(\{I,J\}) \, p(g) \, p(h) \left( g_I h_J + g_J h_I - g - h \right) .
\end{equation}
% By scalar multiplication with $g \in \GGG$, we obtain:
% \begin{equation} \label{eq:dynpg}
% \dd{p(g)}{t} =  \sum_{h \in \GGG} \, \sum_{\{I,J\} \in \PPP} c(\{I,J\}) \left( p(g_I h_J) \, p(g_J h_I)  - p(g) \, p(h) \right) .
% \end{equation}
The contribution of a particular recombination~\eqref{eq:recombination} is identically zero
if $\{g,h\} = \{g_I h_J,g_J h_I\}$ or $c(\{I,J\})=0$.
On the other hand,
different recombination patterns may yield the same recombination (except for the rate constants).
As detailed in Subsection~\ref{subsec:chemreac},
the effect of all patterns causing recombination~\eqref{eq:recombination}
can be summarized in the chemical reaction~\eqref{eq:reaction},
provided that the recombination is effective and the cumulative rate constant is positive.
Reaction~\eqref{eq:reaction} occurs at the rate $\cum{c}(\{I,J\}_{\Delta(g,h)}) \, p(g) \, p(h)$.
In chemical terms,
it follows mass-action-kinetics.

\subsubsection*{Mass-action kinetics}

Let $(\SSS,\CCC,\RRR)$ be a chemical reaction network
and $k \in \Rp^\RRR$ a vector of rate constants.
Under the assumption of mass-action kinetics,
the rate of a reaction $(y \to y') \in \RRR$, which depends on the species concentrations $x \in \Rnn^\SSS$,
is given by $k(y \to y') \, x^y$,
that is,
by a monomial in the reactant concentrations with the corresponding stoichiometric coefficients as exponents.                                           

Hence, we obtain a dynamical system equivalent to~\eqref{eq:dynIJ},
by summing over all reactions~\eqref{eq:R} and assuming mass-action kinetics:
\begin{equation} \label{eq:dyn}
\dd{p}{t} = \sum_{(g+h \to g'+h') \in \RRR} k(g+h \to g'+h') \, p(g) \, p(h)
\left( g' + h' - g - h \right) .
\end{equation}
The right-hand side of~\eqref{eq:dyn} can be written as a product of the stoichiometric matrix
$N \in \R^{\GGG \times \RRR}$ and the rate vector $v_k(p) \in \Rnn^\RRR$.
Thereby, the column vector of $N$ indexed by $(g+h \to g'+h') \in \RRR$ is given by $(g' + h' - g - h) \in \R^\GGG$
and the component of $v_k(p)$ indexed by $g+h \to g'+h'$ is given by $k(g+h \to g'+h') \, p(g) \, p(h)$.
Hence,
\begin{equation} \label{eq:dynNv}
\dd{p}{t} = N v_k(p) .
\end{equation}

\subsubsection*{Complex balancing}

The right-hand side of the dynamical system~\eqref{eq:dyn} can also be written as a product of the complex matrix $Y \in \R^{\GGG \times \CCC}$,
the Laplacian matrix $A_k \in \R^{\CCC \times \CCC}$ of the weighted directed graph,
and the vector of monomials $p^Y \in \R^\CCC$.
The column vector of $Y$ indexed by $y \in \CCC$ is given by $y \in \Rnn^\GGG$ itself,
that is, $Y^y = y$,
and $A_k$ is defined as follows:
$(A_k)_{y'y} = k_{y \to y'}$ if $(y \to y') \in \RRR$,
$(A_k)_{yy} = - \sum_{(y \to y') \in \RRR} k_{y \to y'}$,
and $(A_k)_{y'y} = 0$ otherwise.
We obtain
\begin{equation} \label{eq:dynYA}
\dd{p}{t} = Y A_k \, p^Y 
\end{equation}
Recall that $(p^Y)_y = p^{Y^y} = p^y$ for $y \in \CCC$.
For a particular complex $y=g+h$, we have $p^y = p^{g+h} = p(g) \, p(h)$.

An equilibrium of~\eqref{eq:dynYA} is called complex-balancing if % $p \in \Rp^\GGG$ and
$A_k \, p^Y = 0$.
That is, if at each complex the rates of all reactions sum up to zero.
% If there exists a complex-balancing equilibrium for given $k \in \Rp^\RRR$,
% then all equilibria are complex-balancing, and there is a unique complex-balancing equilibrium in every stoichiometric class.

\subsubsection*{Detailed balancing}

In the process of genetic recombination, all reactions are reversible.
Moreover, the rate constants of a reaction and its reverse coincide.
Hence, we obtain a dynamical system equivalent to~\eqref{eq:dyn},
by summing over all reversible reactions~\eqref{eq:RR}:
\small
\begin{equation} \label{eq:dynrev}
\dd{p}{t} =
\sum_{(g+h \lr g'+h') \in \RRR_\lr} k(g+h \lr g'+h')
\left( p(g) \, p(h) - p(g') \, p(h') \right)
\left( g' + h' - g - h \right) .
\end{equation}
\normalsize
An equilibrium of \eqref{eq:dynrev} is called detailed-balancing
if $p(g) \, p(h) = p(g') \, p(h')$ for all $(g+h \lr g'+h') \in \RRR_\lr$.
In general, an equilibrium of a reversible reaction network is called detailed-balancing
if the rates of each reaction and its reverse coincide.
Clearly, every detailed-balancing equilibrium is complex-balancing.

\subsection{Conserved quantities}

The change over time \eqref{eq:dynNv} lies in a subspace of $\R^\GGG$,
and every trajectory in~$\Rnn^\GGG$ lies in a coset of this subspace. 
We define the stoichiometric subspace
\[
S = \im N
\]
and the stoichiometric compatibility classes
\[
S(p) = (p + S) \cap \Rnn^\GGG 
\]
for $p \in \Rnn^\GGG$.
Every stoichiometric class is characterized by its orthogonal projection on $S^\bot = (\im N)^\bot = \ker N^T$,
that is, by a vector of conserved quantities.
For $u \in S^\bot$, that is, $\con^T N =  0$, we have
\[
\dd{(\con^T p)}{t} = 0 ,
\]
that is, $\con^T p = \textit{const}$.

We observe that the vector $\one \equiv 1^\GGG = \sum_{g \in \GGG} g$ is orthogonal to all columns of~$N$:
$\one^T (g' + h' - g - h) = 0$ for all $(g+h \to g'+h') \in \RRR$,
that is, $\one^T N = 0$.
Since $\one^T p = \sum_{g \in \GGG} p(g)$, we have
\[
\dd{(\sum_{g \in \GGG} p(g))}{t} = 0 ,
\]
and, as one consequence, the simplex $S_\GGG$ is invariant.

Further,
we consider for each locus and each allele at this locus the subset of gametes which contain this allele
and define the corresponding formal sum of gametes
\[
\con_i(a) = \sum_{g \in \GGG:g_i = a} g \quad \text{for } i \in \LLL \text{ and } a \in \AAA_i ,
\]
where $\con_i(a) \in \{0,1\}^\GGG$.
As already mentioned, only complexes which contain the same alleles are connected by a reaction.
Hence, $\con_i(a)^T (g' + h' - g - h) = 0$ for all $(g+h \to g'+h') \in \RRR$,
that is,
$\con_i(a)^T N = 0$,
and the marginal frequencies
\[
p_i(a) = \con_i(a)^T p = \sum_{g \in \GGG:g_i = a} p(g)
\]
are conserved quantities,
that is,
\[
\dd{p_i(a)}{t} = 0 .
\]
For each $i \in \LLL$, we have
\[
\sum_{a \in \AAA_i} \con_i(a) = \sum_{g \in \GGG} g % = \one
\]
and
\[
\sum_{a \in \AAA_i} p_i(a) = \sum_{g \in \GGG} p(g) .
\]
Hence, there are at least $1 + \sum_{i \in \LLL} (A_i - 1)$ linearly independent vectors in~$S^\bot$
and as many independent marginals.

We define the marginal compatibility classes
\[
M(p) = \{ p' \in \Rnn^\GGG \mid p'_i(a) = p_i(a) \text{ for } i \in \LLL \text{ and } a \in \AAA_i \}
\]
for $p \in \Rnn^\GGG$.
Clearly, $S(p) \subseteq M(p)$.

%%%%%%%%%%%%%%%%%%%%%%%%%%%%%%%%%%%%%%%%%%%%%%%%%%%%%%%%%%%%%%%%%%%%%%%%%%%%%%%%%%%%%%%%%%%%%%%%%%%%%%%%%%%%%%%%%%%%%%%%%%

\section{Results}

We determine the equilibria for the process of genetic recombination
and prove convergence to a unique equilibrium.

First, we rewrite the dynamical system~\eqref{eq:dynIJ}.
Using the symmetry in the double sum over recombination partners, we obtain
\begin{align*}
\dd{p}{t} &= \sum_{g,h \in \GGG} \, \sum_{\{I,J\} \in \PPP} c(\{I,J\}) \, p(g) \, p(h) \left( g_I h_J - g \right) \\
&= \sum_{\{I,J\} \in \PPP} c(\{I,J\}) \left( \sum_{g,h \in \GGG} p(g) \, p(h) \, g_I h_J - p \right) .
\end{align*}
Thereby, we assumed $\sum_{g \in \GGG} p(g) = 1$, that is, $p \in S_\GGG$.

For $K \subseteq \LLL$, we define the set of subgametes $\GGG_K = \prod_{i \in K} \AAA_i$,
the projection $\GGG \to \GGG_K, \, g \mapsto g_K$, where $(g_K)_i = g_i$ for $i \in K$,
and its linear extension to the corresponding vector spaces:
$\Rnn^\GGG \to \Rnn^{\GGG_K}, \, p \mapsto p_K$,
where $p_K = \sum_{g \in \GGG} p(g) \, g_K$, that is, $p_K(g_K) = \sum_{h \in \GGG:h_K = g_K} p(h)$.
If $K = \{i\}$ with $i \in \LLL$,
we recover the marginal frequencies
$p_i = \sum_{g \in \GGG} p(g) \, g_i$, that is, $p_i(g_i) = \sum_{h \in \GGG:h_i = g_i} p(h)$.

Let $\{I,J\} \in \PPP$.
For $g \in \GGG_I$ and $h \in \GGG_J$, we define $gh \in \GGG$ as $(gh)_i = g_i$ for $i \in I$ and $(gh)_i = h_i$ for $i \in J$
and extend the multiplication $\GGG_I \times \GGG_J \to \GGG$ linearly to $\Rnn^{\GGG_I} \times \Rnn^{\GGG_J} \to \Rnn^\GGG$.
Hence, we write
\begin{align} \label{eq:dynIJvar}
\dd{p}{t} &= \sum_{\{I,J\} \in \PPP} c(\{I,J\}) \left( \sum_{g \in \GGG} p(g) \, g_I \sum_{h \in \GGG} p(h) \, h_J - p \right) \\
&= \sum_{\{I,J\} \in \PPP} c(\{I,J\}) \left( p_I \, p_J - p \right) . \nonumber
\end{align}
In fact, we may sum over $\{I,J\} \in \PPP^*$ since the contribution of $\{\emptyset,\LLL\}$ is identically zero.

The projection of a trajectory of the dynamical system
is the trajectory of a projected dynamical system with the same structure:
For $K \subseteq \LLL$ and $\{I',J'\}=\{I,J\}_K \in \PPP_K$, we find
\[
(p_I \, p_J)_K = \sum_{g,h \in \GGG} p(g) \, p(g) \, (g_I h_J)_K = \sum_{g,h \in \GGG} p(g) \, p(g) \, g_{I'} h_{J'} = p_{I'} \, p_{J'}
\]
and hence
\begin{align*}
\dd{p_K}{t} &= \sum_{\{I',J'\} \in \PPP_K} \sum_{\substack{ \{I,J\} \in \PPP: \\ \{I,J\} \ge \{I',J'\}}}
c(\{I,J\}) \left( (p_I \, p_J)_K - p_K \right) \\
&= \sum_{\{I',J'\} \in \PPP_K} \cum{c}(\{I',J'\}) \left( p_{I'} \, p_{J'} - p_K \right) .
\end{align*}

\subsection{Equilibria}

Now, we are in a position to characterize the equilibria on the simplex.
\begin{lem} \label{lem1}
For all recombination rate constants,
$p \in S_\GGG$ is an equilibrium of the dynamical system~\eqref{eq:dynIJ} if and only if
\begin{equation} \label{eq:equ}
p = \prod_{i \in \LLL} p_i ,
\quad \text{that is,} \quad
p(g) = \prod_{i \in \LLL} p_i(g_i) .
\end{equation}
\end{lem}
\begin{proof}
We show that, if~\eqref{eq:equ}, then
\[
p_I = \prod_{i \in I} p_i .
\]
Indeed, for $\{I,J\} \in \PPP$ and $J = \{1, \ldots, |J|\}$, we find
\begin{align*}
p_I (g_I) &= \sum_{h_J \in \GGG_J} p(g_I h_J) \\
&= \sum_{h_1 \in \GGG_1} \ldots \sum_{h_\J \in \GGG_\J} \prod_{i \in I} p_i(g_i) \prod_{i \in J} p_i(h_i) \\
&= \prod_{i \in I} p_i(g_i) \sum_{h_1 \in \GGG_1} p_1(h_1) \; \ldots \sum_{h_\J \in \GGG_\J} p_\J(h_\J) \\
&= \prod_{i \in I} p_i(g_i) .
\end{align*}
Hence, 
\[
p_I \, p_J = \prod_{i \in I} p_i \, \prod_{i \in J} p_i = \prod_{i \in \LLL} p_i = p
\]
for all $\{I,J\} \in \PPP$,
and $p \in S_\GGG$ is an equilibrium of the dynamical system~\eqref{eq:dynIJvar}
equivalent to~\eqref{eq:dynIJ}.

It remains to show that $p \in S_\GGG$ is an equilibrium only if~\eqref{eq:equ}.
We proceed by induction on the number of loci:

For $\LLL=\{1\}$,
there is no non-trivial recombination.
Every $p \in S_\GGG$ is an equilibrium which coincides with its marginals:
$p(g) = p_1(g)$ for $g \in \GGG$,
that is, $p = p_1$.

For $L\ge2$,
we consider subsets of loci $K \subset \LLL$ with $|K|<L$.
The projection of an equilibrium $p \in S_\GGG$ of the dynamical system (with loci $\LLL$)
is an equilibrium of the projected dynamical system (with loci $K$).
By the induction hypothesis, $p_K = \prod_{i \in K} p_i$ and hence
\[
p_I \, p_J = \prod_{i \in I} p_i \, \prod_{i \in J} p_i = \prod_{i \in \LLL} p_i
\]
for all $\{I,J\} \in \PPP^*$.
Summing over $\{I,J\} \in \PPP^*$ in~\eqref{eq:dynIJvar}, we obtain
\[
0 = \sum_{\{I,J\} \in \PPP^*} c(\{I,J\}) \left( \prod_{i \in \LLL} p_i - p \right)
\]
and hence $p = \prod_{i \in \LLL} p_i$.
(There always exists a non-zero rate constant for some non-trivial recombination pattern.)
\end{proof}

If the dynamics is not restricted to the simplex,
then $p \in \Rnn^\GGG$ is an equilibrium if and only if
\[
p = \left( \ts \sum_{g \in \GGG} p(g) \right)^{1-L} \prod_{i \in \LLL} p_i .
\]

Clearly, every equilibrium is determined by its marginals,
and we have the following result.
\begin{pro} \label{pro1}
Every marginal compatibility class contains a unique equilibrium.
% (For $p' \in S_\GGG$, the equilibrium $p \in M(p')$ is given by $p = \prod_{i \in \LLL} p'_i$.)
\end{pro}

In chemical terms,
every equilibrium is detailed-balancing since
\[
p(g) \, p(h) = \prod_{i \in \LLL} p_i(g_i) \, p_i(h_i) = \prod_{i \in \LLL} p_i(g'_i) \, p_i(h'_i) = p(g') \, p(h')
\]
for all $(g+h \lr g'+h') \in \RRR_\lr$, cf.\ Equation~\eqref{eq:dynrev}.
Recall that only complexes which contain the same alleles are connected by a reaction.
In fact, we can derive the following result entirely in the chemical setting.
\begin{pro} \label{pro2}
Every stoichiometric compatibility class contains a unique positive equilibrium, which is detailed-balancing,
and no boundary equilibria.
\end{pro}
\begin{proof}
First, we determine the set of positive detailed-balancing equilibria.
Using positivity, we write the condition for detailed balancing,
\[
p(g) \, p(h) = p(g') \, p(h') \quad \text{for all } (g+h \lr g'+h') \in \RRR_\lr ,
\]
as
\[
p^{g'+h'-g-h} = 1 \quad \text{for all } (g+h \to g'+h') \in \RRR
\]
and even more abstractly as
\[
p^N = 1 ,
\]
where $N \in \R^{\GGG \times \RRR}$ is the stoichiometric matrix and $1 \equiv 1^{\!\RRR}$.
Clearly, the trivial solution is given by $p^* = 1 \equiv 1^\GGG$.
To determine all solutions,
we take the logarithm,
\[
N^T \ln p = 0 ,
\]
and note that $\ker N^T = (\im N)^\bot = S^\bot$ and $\dim(S^\bot) \ge 1$. % (since $N^T 1 = 0$).
Hence, we can write the set of positive detailed-balancing equilibria as
\[
Z = \{ p \in \Rp^\GGG \mid \ln p - \ln p^* \in S^\bot \} .
\]
Clearly, every detailed-balancing equilibrium is complex-balancing.
Now, if there exists a positive complex-balancing equilibrium $p^*$,
then the set of positive complex-balancing equilibria is given by $Z$
and there are no other positive equilibria \cite[Theorem 6A]{HornJackson1972}.
Moreover, every stoichiometric compatibility class contains a unique positive equilibrium \cite[Lemma 4B]{HornJackson1972}.
Hence, the sets of positive detailed- and complex-balancing equilibria coincide,
and every stoichiometric compatibility class contains a unique positive equilibrium, which is detailed-balancing.

It remains to preclude boundary equilibria.
We consider an arbitrary initial value $p \in \Rnn^\GGG$ on the boundary,
that is, $p(g) = 0$ for some $g \in \GGG$,
and define the set of gametes $\GGG_0 = \{ g \in \GGG \mid p(g) > 0 \}$.
Note that for each locus and each allele at this locus
there is a gamete which contains this allele and occurs with a positive frequency.
Further, recall that each pair of loci gets separated by some recombination pattern with positive rate constant.
By Lemma~\ref{lem2} below, the set of all gametes $\GGG$ is reachable from $\GGG_0$.
Now, let $p(t)$ be the solution of the dynamical system with $p(0)=p$.
By \cite[Theorem 2, p.\ 618]{Volpert1985}, $p(t) \in \Rp^\GGG$ for $t>0$, and hence $p$ is not an equilibrium.
\end{proof}

In the proof,
we have used fundamental results about complex balancing by Horn and Jackson~\cite{HornJackson1972}.
Necessary and sufficient conditions for complex balancing are given by Horn~\cite{Horn1972}
and for detailed balancing by Vol'pert and Hudjaev~\cite{Volpert1985} and Feinberg~\cite{Feinberg1989}.
For the relation between complex and detailed balancing,
see~\cite{DickensteinPerezMillan2011}.

In the proof of Proposition~\ref{pro2}, we have also used the purely graph-theoretical concept of reachability.
Let $\SSS$ and $\RRR$ be the species and reactions of a chemical reaction network,
and let $\SSS_0 \subseteq \SSS$ be a set of species.
Iteratively, we define
\[
\SSS_i = \SSS_{i-1} \cup \{ g' \mid g,h \in \SSS_{i-1} \text{ and } (g+h \to g'+h') \in \RRR \}
\]
for $i \ge 1$.
Since the graph is finite, we find $\SSS_i = \SSS_{i^*}$ for some $i^* \ge 0$ and all $i \ge i^*$,
and the set of species reachable from $\SSS_0$ is given by $\SSS_{i^*}$.

For chemical reaction networks arising from genetic recombination, we have the following result.
\begin{lem} \label{lem2}
In a process of recombination,
assume that every pair of loci gets separated by some pattern with positive rate constant.
In the resulting chemical reaction network,
every gamete is reachable from a given set of gametes,
provided that every allele is contained in some gamete in this set.
\end{lem}
\begin{proof}
We use induction on the number of loci:

For $L=1$, the gametes coincide with the alleles.

For $L \ge 2$, 
we consider subsets of loci $K \subset \LLL$ with $|K|<L$
and project the network and the given set of gametes on the loci $K$.
In the projected network,
every pair of loci gets separated by some pattern,
and in the projected set,
every allele (at loci $K$) is contained in some gamete.
By the induction hypothesis,
every gamete $g_K \in \GGG_K$ is reachable in the projected network,
and hence some gamete $h \in \GGG$ with $h_K = g_K$ is reachable.

It remains to show that every gamete $g' \in \GGG$ is reachable.
Let $\LLL = \{1,\ldots,L\}$ and $G = \LLL \setminus \{1\}$, $H = \LLL \setminus \{2\}$.
By the argument above, some gametes $g,h \in \GGG$ with $g_G = g'_G$, $h_H = g'_H$ are reachable.
If $g'$ equals either $g$ or $h$, then it is reachable.
Otherwise, since loci $1$ and $2$ get separated by some pattern, we find the reaction $g + h \to g' + h'$, and hence $g'$ is reachable.
\end{proof}

On the one hand, by Proposition~\ref{pro1},
every marginal compatibility class contains a unique equilibrium.
Hence, the set of all equilibria can be parame\-trized by $1 + \sum_{i \in \LLL} (A_i - 1)$ independent marginals.
On the other hand, by Proposition~\ref{pro2},
every stoichiometric compatibility class contains a unique equilibrium.
Since stoichiometric classes are contained in marginal classes,
we have the following result.
\begin{cor} \label{cor}
The stoichiometric compatibility classes coincide with the marginal compatibility classes,
and $\dim(S^\bot) = 1 + \sum_{i \in \LLL} (A_i - 1)$.
\end{cor}

\subsection{Convergence}

Our main results concern the convergence of the dynamics to a unique equilibrium.
For the first theorem, we provide two proofs:
one by induction (as in the original literature) and one using the entropy as a Lyapunov function.
For the second theorem, formulated in the chemical setting,
we rely on results from chemical reaction network theory
which are based on the same Lyapunov function.

\begin{thm} \label{thm:genetics}
In every marginal compatibility class and for all recombination rate constants,
a process of genetic recombination converges to the unique equilibrium given by \eqref{eq:equ}.
\end{thm}
\begin{proof}[First Proof (Induction)]
Every marginal compatibility class is characterized by a unique equilibrium.
Given an equilibrium $p \in S_\GGG$, that is, $p = \prod_{i \in \LLL} p_i$,
we consider trajectories in the class $M(p)$.
We proceed by induction on the number of loci:

For $L = 1$, we have $M(p)=p$.

For $L \ge 2$,
we consider subsets of loci $K \subset \LLL$ with $|K|<L$.
The projection of a trajectory $\phi: \Rnn \to M(p)$ of the dynamical system (with loci $\LLL$)
is a trajectory of the projected dynamical system (with loci $K$).
By the induction hypothesis, $\phi(t)_K \to \prod_{i \in K} p_i$ as $t \to \infty$ and hence
\[
\phi(t)_I \, \phi(t)_J \to \prod_{i \in I} p_i \, \prod_{i \in J} p_i = \prod_{i \in \LLL} p_i = p
\]
for all $\{I,J\} \in \PPP^*$.
Summing over $\{I,J\} \in \PPP^*$ in the dynamical system~\eqref{eq:dynIJvar} equivalent to~\eqref{eq:dynIJ},
we obtain the non-autonomous differential equation
\begin{align*}
\dd{\phi}{t} &= \sum_{\{I,J\} \in \PPP^*} c(\{I,J\}) \left( \phi_I \, \phi_J - \phi \right) \\
&= f(t) - \sum_{\{I,J\} \in \PPP^*} c(\{I,J\}) \, \phi \nonumber
\end{align*}
with 
\[
f(t) = \sum_{\{I,J\} \in \PPP^*} c(\{I,J\}) \, \phi_I \, \phi_J
\]
and
\[
f(t) \to \sum_{\{I,J\} \in \PPP^*} c(\{I,J\}) \, p
\]
as $t \to \infty$.
In other words, the differential equation is asymptotically autonomous.
The limiting equation
\[
\dd{\phi}{t} = \sum_{\{I,J\} \in \PPP^*} c(\{I,J\}) \left( p - \phi \right)
\]
is linear,
and hence $\phi(t) \to p$ in the limiting equation.
Moreover, $\{p\}$ is the maximal compact invariant set in the limiting system,
and therefore $\phi(t) \to p$ holds also for all solutions of the original dynamical system, see e.g. \cite{markus1956, mischaikow1995}.
\end{proof}

\begin{proof}[Second Proof (Lyapunov function)]
We consider the classical entropy function
\begin{equation*} \label{eq:entropy}
H(p) = - \sum_{g \in \GGG} p(g) \ln p(g) = - \, p^T \ln p \geq 0 
\end{equation*}
which defines a continuous function on the simplex $S_\GGG$.
If $p(g) > 0$ for all $g \in \GGG$, then $H$ is smooth and
\begin{equation*} \label{eq:entropyd1}
\dot H(p) = - \sum_{g \in \GGG} \dot p(g) \ln p(g) - \sum_{g \in \GGG} \dot p(g) = - \, \dot p^T \ln p ,
\end{equation*}
since $\sum_{g \in \GGG} p(g) = 1$.
Using the dynamical system~\eqref{eq:dynrev} equivalent to~\eqref{eq:dynIJ}, we obtain
\begin{align*} \label{eq:entropyd2}
\dot H(p) % &= - \langle \dot p, \ln p \rangle \nonumber \\
&= \sum_{(g+h \lr g'+h') \in \RRR_\lr} k(g+h \lr g'+h') \left( p(g) \, p(h) - p(g') \, p(h') \right) \cdot \nonumber \\
& \qquad \qquad \qquad \qquad \cdot \left( \ln (p(g) \, p(h)) - \ln (p(g') \, p(h')) \right) \geq 0 .
\end{align*}
Equality $\dot H(p) = 0$ holds if and only if $p$ is a detailed-balancing equilibrium,
that is, if and only if~\eqref{eq:equ} holds.

Given an initial point $p(0) \in \Rp^{\GGG}$ in the interior,
the entropy $H(p(t))$ increases strictly towards its maximum on $M(p(0))$,
and $p(t)$ converges to the unique equilibrium $p$ in the class $M(p(0))$.

Given an initial point $p(0) \in \Rnn^{\GGG}$ on the boundary,
we have $p(t) \in \Rp^{\GGG}$ for $t > 0$, cf.\ the proof of Proposition~\ref{pro2}.
In genetic terms,
recombination immediately produces all gametes, as long as all alleles are present in the population.
% If some marginal is zero, that is, some allele is not present, then we ignore this allele. 
\end{proof}

The entropy as Lyapunov function was used by Akin~\cite{akin1979} and Lyubich~\cite{lyubich1992}
(referring to a paper by Kun and Lyubich~\cite{kun1980})
to prove global stability for recombination.
For chemical reaction networks with detailed balancing, see Vol'pert and Hudjaev~\cite{Vasilev1974,Volpert1985}
(who acknowledge previous use of the entropy function by Zel'dovich~\cite{Zeldovich1938}).
For complex balancing, see~\cite{higgins1968,HornJackson1972,Siegel2000,gorban2014}. 

\begin{thm} \label{thm:chemistry}
A process of genetic recombination gives rise to a reversible chemical reaction network with mass-action kinetics.
In every stoichiometric compatibility class and for all reaction rate constants,
the dynamics converges to a unique positive detailed-balancing equilibrium.
\end{thm}
\begin{proof}
By \cite[Theorem, pp.\ 642--643]{Volpert1985}, the $\omega$-limit set of every solution
consists either of a unique positive detailed-balancing equilibrium or boundary detailed-balancing equilibria.
By Proposition~\ref{pro2},
there are no boundary equilibria,
and every solution converges to a unique positive detailed-balancing equilibrium.
\end{proof}

%%%%%%%%%%%%%%%%%%%%%%%%%%%%%%%%%%%%%%%%%%%%%%%%%%%%%%%%%%%%%%%%%%%%%%%%%%%%%%%%%%%%%%%%%%%%%%%%%%%%%%%%%%%%%%%%%%%%%%%%%%

\section{Final remarks}

Note that we have not used a central concept of chemical reaction network theory, the deficiency
\[
\delta = m - l - s ,
\]
where $m$ is the number of complexes, $l$ the number of linkage classes, and $s$ the dimension of the stoichiometric subspace.

The deficiency zero and one theorems state that
there exists a unique (asymptotically stable) positive complex-balancing equilibrium,
for all reaction rate constants and all stoichiometric compatibility classes,
if the network is weakly reversible and either (0) its deficiency is zero
or (1a) the deficiencies of the individual linkage classes are zero or one
and (1b) the individual deficiencies add up to the total deficiency,
see~\cite{Feinberg1995a}.

In Example~1 ($L=2$, $A_1=A_2=2$), we find $\delta = 2 - 1 - 1 = 0$.
However, already in Example~2 ($L=3$, $A_1=A_2=A_3=2$),
the deficiencies of the individual linkage classes are zero, but $\delta = 16 - 7 - 4 = 5$.

In fact, the individual deficiencies are zero in the entire Scheme~3 ($L \ge 2$, $A_i=2$ for $i=1,\ldots,L$):
Every linkage class is characterized by two different alleles at some loci $K \in \LLL$ with $|K| \ge 2$
and two identical alleles at other loci.
Hence $2^{|K|-1} - 1 - (2^{|K|-1} - 1) = 0$, whereas
\begin{align*}
\delta &= m - l - s \\
&= 2^{L-1} (2^L - (1+L)) - (3^L - 2^{L-1} (2+L)) - (2^L - (1+L)) \\
&= 2^{L-1} (2^L - 1) - 3^L + 1 + L ,
\end{align*}
using $s = \dim(S) = |\GGG| - \dim(S^\bot) = 2^L - (1+L)$.
For $L=3,4,5,...\,$, we find $\delta = 5,44,259,...\,$, and the deficiency zero and one theorems do not apply.

More importantly,
there exist $\delta$ necessary and sufficient conditions on the rate constants
for the existence of complex-balancing equilibria, see Horn~\cite{Horn1972}.
The conditions involve the Laplacian matrix of the weighted directed graph of complexes and reactions,
in particular, the quotients of so-called tree constants, see~\cite{MuellerRegensburger2014}.
For the existence of detailed-balancing equilibria,
additionally the Wegscheider conditions have to be fulfilled,
that is, the products of rate constants in a cycle and its reverse must coincide, see~\cite{Wegscheider1901,Volpert1985,Feinberg1989}.
In the process of genetic recombination,
the rate constants of a reaction and its reverse coincide,
and all conditions for the existence of complex- and detailed-balancing equilibria are fulfilled.

%%%%%%%%%%%%%%%%%%%%%%%%%%%%%%%%%%%%%%%%%%%%%%%%%%%%%%%%%%%%%%%%%%%%%%%%%%%%%%%%%%%%%%%%%%%%%%%%%%%%%%%%%%%%%%%%%%%%%%%%%%

\bibliographystyle{siam}
\bibliography{crnt,recom}

\end{document}